\documentclass[pra,notitlepage,twocolumn,nofootinbib,superscriptaddress,longbibliography]{revtex4-2}

\usepackage{mathtools}
\usepackage{amsmath}
\usepackage[shortlabels]{enumitem}
\usepackage[dvipsnames]{xcolor}
\usepackage{framed}
\definecolor{shadecolor}{rgb}{0.9,0.9,0.9}
\usepackage{graphicx,epic,eepic,epsfig,amsmath,latexsym,amssymb,verbatim,color}
 
\usepackage{amsfonts}       
\usepackage{nicefrac}       
\usepackage{mathrsfs}

\usepackage{amsmath}
\usepackage{bbm}

\usepackage{float}
\usepackage{tikz}
\usetikzlibrary{chains}
\usetikzlibrary{fit}
\usetikzlibrary{arrows}
\usetikzlibrary{decorations}

\usepackage{epsfig}
\usetikzlibrary{shapes.symbols,patterns} 
\usepackage{pgfplots}

\usepackage[strict]{changepage}
\usepackage{hyperref}
\hypersetup{colorlinks=true,citecolor=blue,linkcolor=blue,filecolor=blue,urlcolor=blue,breaklinks=true}

\usepackage[marginal]{footmisc}
\usepackage{url}
\usepackage{theorem}

\newtheorem{definition}{Definition}
\newtheorem{proposition}{Proposition}
\newtheorem{lemma}[proposition]{Lemma}

\newtheorem{theorem}[proposition]{Theorem}

\def\squareforqed{\hbox{\rlap{$\sqcap$}$\sqcup$}}
\def\qed{\ifmmode\squareforqed\else{\unskip\nobreak\hfil
\penalty50\hskip1em\null\nobreak\hfil\squareforqed
\parfillskip=0pt\finalhyphendemerits=0\endgraf}\fi}
\def\endenv{\ifmmode\;\else{\unskip\nobreak\hfil
\penalty50\hskip1em\null\nobreak\hfil\;
\parfillskip=0pt\finalhyphendemerits=0\endgraf}\fi}
\newenvironment{proof}{\noindent \textbf{{Proof~} }}{\hfill $\blacksquare$}

\newcounter{remark}

\newcounter{example}

\mathchardef\ordinarycolon\mathcode`\:
\mathcode`\:=\string"8000
\def\vcentcolon{\mathrel{\mathop\ordinarycolon}}
\begingroup \catcode`\:=\active
  \lowercase{\endgroup
  \let :\vcentcolon
  }

\usepackage{cleveref}
\usepackage{graphicx}
\usepackage{xcolor}

\RequirePackage[framemethod=default]{mdframed}
\newmdenv[skipabove=7pt,
skipbelow=7pt,
backgroundcolor=darkblue!15,
innerleftmargin=5pt,
innerrightmargin=5pt,
innertopmargin=5pt,
leftmargin=0cm,
rightmargin=0cm,
innerbottommargin=5pt,
linewidth=1pt]{tBox}

\newmdenv[skipabove=7pt,
skipbelow=7pt,
backgroundcolor=darkred!15,
innerleftmargin=5pt,
innerrightmargin=5pt,
innertopmargin=5pt,
leftmargin=0cm,
rightmargin=0cm,
innerbottommargin=5pt,
linewidth=1pt]{rBox}

\newmdenv[skipabove=7pt,
skipbelow=7pt,
backgroundcolor=blue2!25,
innerleftmargin=5pt,
innerrightmargin=5pt,
innertopmargin=5pt,
leftmargin=0cm,
rightmargin=0cm,
innerbottommargin=5pt,
linewidth=1pt]{dBox}
\newmdenv[skipabove=7pt,
skipbelow=7pt,
backgroundcolor=darkkblue!15,
innerleftmargin=5pt,
innerrightmargin=5pt,
innertopmargin=5pt,
leftmargin=0cm,
rightmargin=0cm,
innerbottommargin=5pt,
linewidth=1pt]{sBox}
\definecolor{darkblue}{RGB}{0,76,156}
\definecolor{darkkblue}{RGB}{0,0,153}
\definecolor{blue2}{RGB}{102,178,255}
\definecolor{darkred}{RGB}{195,0,0}

\newcommand{\nc}{\newcommand}
\nc{\rnc}{\renewcommand}
\nc{\lbar}[1]{\overline{#1}}
\nc{\bra}[1]{\langle#1|}
\nc{\ket}[1]{|#1\rangle}
\nc{\ketbra}[2]{|#1\rangle\!\langle#2|}
\nc{\braket}[2]{\langle#1|#2\rangle}

\nc{\proj}[1]{| #1\rangle\!\langle #1 |}
\nc{\avg}[1]{\langle#1\rangle}
\nc{\rank}{\operatorname{Rank}}
\nc{\smfrac}[2]{\mbox{$\frac{#1}{#2}$}}
\nc{\tr}{\operatorname{Tr}}
\nc{\ox}{\otimes}
\nc{\dg}{\dagger}
\nc{\dn}{\downarrow}
\nc{\cA}{{\cal A}}
\nc{\cB}{{\cal B}}
\nc{\cC}{{\cal C}}
\nc{\cD}{{\cal D}}
\nc{\cE}{{\cal E}}
\nc{\cF}{{\cal F}}
\nc{\cG}{{\cal G}}
\nc{\cH}{{\cal H}}
\nc{\cI}{{\cal I}}
\nc{\cJ}{{\cal J}}
\nc{\cK}{{\cal K}}
\nc{\cL}{{\cal L}}
\nc{\cM}{{\cal M}}
\nc{\cN}{{\cal N}}
\nc{\cO}{{\cal O}}
\nc{\cP}{{\cal P}}
\nc{\cQ}{{\cal Q}}
\nc{\cR}{{\cal R}}
\nc{\cS}{{\cal S}}
\nc{\cT}{{\cal T}}
\nc{\cU}{{\cal U}}
\nc{\cV}{{\cal V}}
\nc{\cX}{{\cal X}}
\nc{\cY}{{\cal Y}}
\nc{\cZ}{{\cal Z}}
\nc{\cW}{{\cal W}}
\nc{\csupp}{{\operatorname{csupp}}}
\nc{\qsupp}{{\operatorname{qsupp}}}
\nc{\var}{{\operatorname{var}}}
\nc{\rar}{\rightarrow}
\nc{\lrar}{\longrightarrow}
\nc{\polylog}{{\operatorname{polylog}}}
\nc{\wt}{{\operatorname{wt}}}
\nc{\av}[1]{{\left\langle {#1} \right\rangle}}
\nc{\supp}{{\operatorname{supp}}}

\nc{\argmin}{{\operatorname{argmin}}}

\def\x{\xi}

\nc{\RR}{{{\mathbb R}}}
\nc{\CC}{{{\mathbb C}}}
\nc{\FF}{{{\mathbb F}}}
\nc{\NN}{{{\mathbb N}}}
\nc{\ZZ}{{{\mathbb Z}}}
\nc{\PP}{{{\mathbb P}}}
\nc{\QQ}{{{\mathbb Q}}}
\nc{\UU}{{{\mathbb U}}}
\nc{\EE}{{{\mathbb E}}}
\nc{\id}{{\operatorname{id}}}

\nc{\CHSH}{{\operatorname{CHSH}}}

\nc{\be}{\begin{equation}}
\nc{\ee}{{\end{equation}}}
\nc{\bea}{\begin{eqnarray}}
\nc{\eea}{\end{eqnarray}}
\nc{\<}{\langle}
\rnc{\>}{\rangle}
\nc{\rU}{\mbox{U}}

\nc{\ob}[1]{#1}

\nc{\OLOCC}{{\text{1-LOCC}}}
\nc{\SEP}{{\text{SEP}}}
\nc{\NS}{{\text{NS}}}
\nc{\LOCC}{{\text{LOCC}}}
\nc{\PPT}{{\text{PPT}}}
\nc{\EXT}{{\text{EXT}}}
\nc{\Sym}{{\operatorname{Sym}}}

\nc{\ERLO}{{E_{{ R,LO}}}}
\nc{\ERLOCC}{{E_{{R,\text{PPT}}}}}
\nc{\ERPPT}{{E_{{R,\text{PPT}}}}}
\nc{\ERPPTinf}{{E^{\infty}_{{R,\text{PPT}}}}}
\nc{\ER}{E_{\rm R}}
\nc{\ERLOCCinfty}{{E^{\infty}_{{r,LOCC}}}}
\nc{\Aram}{{\operatorname{\sf A}}}
\nc{\ECPPT}{{E_{{C,\text{PPT}}}}}
\nc{\EDPPT}{{E_{{D,\text{PPT}}}}}
\nc{\Freek}{{\text{PPT$_k$}}}
\nc{\Freesec}{{\text{PPT$_2$}}}

\nc{\NB}{{{{\tiny N}}}}
\nc{\LB}{{{LN}}}
\nc{\NPT}{{\text{NPT}}}

\usepackage{tikz}

\makeatletter
\def\grd@save@target#1{%
  \def\grd@target{#1}}
\def\grd@save@start#1{%
  \def\grd@start{#1}}
\tikzset{
  grid with coordinates/.style={
    to path={%
      \pgfextra{%
        \edef\grd@@target{(\tikztotarget)}%
        \tikz@scan@one@point\grd@save@target\grd@@target\relax
        \edef\grd@@start{(\tikztostart)}%
        \tikz@scan@one@point\grd@save@start\grd@@start\relax
        \draw[minor help lines,magenta] (\tikztostart) grid (\tikztotarget);
        \draw[major help lines] (\tikztostart) grid (\tikztotarget);
        \grd@start
        \pgfmathsetmacro{\grd@xa}{\the\pgf@x/1cm}
        \pgfmathsetmacro{\grd@ya}{\the\pgf@y/1cm}
        \grd@target
        \pgfmathsetmacro{\grd@xb}{\the\pgf@x/1cm}
        \pgfmathsetmacro{\grd@yb}{\the\pgf@y/1cm}
        \pgfmathsetmacro{\grd@xc}{\grd@xa + \pgfkeysvalueof{/tikz/grid with coordinates/major step}}
        \pgfmathsetmacro{\grd@yc}{\grd@ya + \pgfkeysvalueof{/tikz/grid with coordinates/major step}}
        \foreach \x in {\grd@xa,\grd@xc,...,\grd@xb}
        \node[anchor=north] at (\x,\grd@ya) {\pgfmathprintnumber{\x}};
        \foreach \y in {\grd@ya,\grd@yc,...,\grd@yb}
        \node[anchor=east] at (\grd@xa,\y) {\pgfmathprintnumber{\y}};
      }
    }
  },
  minor help lines/.style={
    help lines,
    step=\pgfkeysvalueof{/tikz/grid with coordinates/minor step}
  },
  major help lines/.style={
    help lines,
    line width=\pgfkeysvalueof{/tikz/grid with coordinates/major line width},
    step=\pgfkeysvalueof{/tikz/grid with coordinates/major step}
  },
  grid with coordinates/.cd,
  minor step/.initial=.2,
  major step/.initial=1,
  major line width/.initial=2pt,
}
\makeatother

\usepackage{thmtools}
\usepackage{thm-restate}
\usepackage{etoolbox}
\makeatletter
\def\problem@s{}
\newcounter{problems@cnt}

\newcommand{\allproblems}{\problem@s}
\makeatother

\usepackage{amsmath,amsfonts,amssymb,array,graphicx,mathtools,multirow,bm,times,tcolorbox,relsize,booktabs}
\usepackage[utf8]{inputenc}
\usepackage[T1]{fontenc}
\usepackage{qcircuit}

\usepackage[ruled, vlined, linesnumbered]{algorithm2e}

\definecolor{colortwo}{rgb}{0.4,0.77,0.17}
\definecolor{colorthree}{rgb}{0.01,0.51,0.93}
\allowdisplaybreaks

\begin{document}
\title{Optimal unilocal virtual quantum broadcasting}
\author{Hongshun Yao}
\thanks{H. Yao and X. Liu contributed equally to this work.}
\author{Xia Liu}
\thanks{H. Yao and X. Liu contributed equally to this work.}
\author{Chengkai Zhu}
\author{Xin Wang}
\email{felixxinwang@hkust-gz.edu.cn}
\affiliation{Thrust of Artificial Intelligence, Information Hub,\\
Hong Kong University of Science and Technology (Guangzhou), Guangdong 511453, China}

\begin{abstract}
Quantum broadcasting is central to quantum information processing and characterizes the correlations within quantum states. Nonetheless, traditional quantum broadcasting encounters inherent limitations dictated by the principles of quantum mechanics. In a previous study, Parzygnat \textit{et al.} \href{https://journals.aps.org/prl/abstract/10.1103/PhysRevLett.132.110203}{[Phys. Rev. Lett. \textbf{132}, 110203 (2024)]} introduced a canonical broadcasting quantum map that goes beyond the quantum no-broadcasting theorem through a virtual process. In this work, we generalize the concept of virtual broadcasting to unilocal broadcasting by incorporating a reference system and introduce protocols that can be approximated using physical operations with minimal cost. First, we propose a universal unilocal protocol enabling multiple parties to share the correlations of a target bipartite state, which is encoded in the expectation value for any observable. Second, we formalize the simulation cost of a virtual quantum broadcasting protocol into a semidefinite programming problem. Notably, we propose a specific protocol with optimal simulation cost for the 2-broadcasting scenario, revealing an explicit relationship between simulation cost and the quantum system's dimension. Moreover, we establish upper and lower bounds on the simulation cost of the virtual $n$-broadcasting protocol and demonstrate the convergence of the lower bound to the upper bound as the quantum system's dimension increases.
\end{abstract}

\date{\today}
\maketitle

\section{Introduction}
In classical information processing, creating duplicates is a straightforward task. However, the quantum realm presents a challenge due to the no-cloning theorem~\cite{Wootters1982,Dieks1982}, rendering direct copies impossible. 
Quantum broadcasting~\cite{Piani2015a,Xie2017}, a concept milder than quantum cloning, offers a distinct perspective on the classical-quantum interface. 
Unfortunately, there are also fundamental restrictions on quantum broadcasting~\cite{Barnum1996}. The no-broadcasting theorem states that it is only possible to broadcast a set of quantum states if they commute with each other. In other words, if the quantum states have properties that can be simultaneously measured without disturbing each other, it is possible to broadcast them.

These no-go theorems can be further extended to the setting of local broadcasting for composite quantum systems \cite{piani2008no,Luo2010,Piani2016a}. Given a bipartite quantum state $\rho_{A B}$ shared by Alice and Bob, the local-broadcasting aims to perform local operations $\Lambda_{A \rightarrow A_1 A_2}$ and $\Gamma_{B \rightarrow B_1 B_2}$ to produce a state $\widehat{\rho}_{A_1 A_2 B_1 B_2}:=\left(\Lambda_{A \rightarrow A_1 A_2} \ox \Gamma_{B \rightarrow B_1 B_2}\right) \rho_{A B}$ such that $\tr_{A_1 B_1} [\widehat{\rho}_{A_1 A_2 B_1 B_2}]=\tr_{A_2 B_2}[ \widehat{\rho}_{A_1 A_2 B_1 B_2}]=\rho_{A B}$. Furthermore, unilocal broadcasting is considered when the local operations are only allowed for one party, e.g., Bob. It is shown that the unilocal broadcasting can be done if and only if $\rho_{A B}$ is classical on $B$~\cite{piani2008no,Piani2016a,luo2010quantum,Luo2010}. More generally, a unilocal $n$-broadcasting performs the local operation $\Gamma_{B\rightarrow B_1\cdots B_n}$ to produce the state $\widehat{\rho}_{AB_1\cdots B_n}:=\Gamma_{B\rightarrow B_1\cdots B_n}(\rho_{AB})$ such that $\tr_{\backslash AB_1}[\widehat{\rho}_{AB_1\cdots B_n}]=\cdots=\tr_{\backslash AB_n}[\widehat{\rho}_{AB_1\cdots B_n}]=\rho_{AB}$~\cite{Xie2017}, which shown in Fig.~\ref{fig:unilocal}.

Although any physical process cannot achieve quantum broadcasting due to these no-go theorems, Parzygnat et al.~\cite{parzygnat2024virtual} presented a canonical broadcasting quantum map going beyond the quantum no-broadcasting theorem via a virtual process, which focuses on broadcasting measurement statistics of a target state rather than the state itself. They presented three natural conditions that virtual broadcasting maps should satisfy and provided several physical interpretations, such as that the universal quantum cloner is the optimal physical approximation to their canonical broadcasting map. However, when considering using physical operators to approximate the non-physical process with minimal sampling cost, the optimal protocols for
virtual broadcasting and the more general unilocal broadcasting are unknown. To overcome the above challenges as well as the limitations of the quantum no-local-broadcasting theorem, we investigate \textit{unilocal virtual quantum broadcasting}, which aims to broadcast the correlation of a bipartite quantum system encoded in the expectation values of any possible observables.

Virtual processes concern the classical information discerned after measurement, referred to as \textit{shadow information}~\cite{huang2020predicting,li2021vsql} that we mainly focus on in the majority of quantum information and quantum computing tasks, rather than the whole information of a state. Therefore, for a bipartite quantum state $\rho_{AB}$, we concentrate on a specific broadcasting task that the local operations employed by Bob enable $n$ parties $B_1,\cdots, B_n$ to access the same shadow information $\tr[O\rho_{AB}]$ with respect to any observable $O$. It is worth noting that we are not focusing on distributing the expectation value as classical bits to different parties. In fact, Alice and Bob are considered geographically separated laboratories where a global expectation value cannot be obtained in the first place. Instead, our framework works by supplying Alice and Bob with multiple identical copies of the states which ensures that each bipartite party $AB_j$ can access the correlation of $AB$ by sharing the same expectation value.

Technically, we extend the traditional quantum broadcasting by employing Hermitian preserving and trace-preserving (HPTP) maps, which can be physically implemented by quasiprobability decomposition (QPD)~\cite{Buscemi2013,Temme2017,Jiang2020,Piveteau2021} and measurement-controlled post-processing~\cite{Zhao2023}. Such physical simulation of unphysical maps plays a crucial role in applications such as entanglement detection~\cite{Peres1996,Horodecki1996,Guhne2009,Wang2020}, error mitigation~\cite{Temme2017,Piveteau2021,Jiang2020,Zhao2022,zhao2024retrieving}, and two-point correlator~\cite{Buscemi2013}. In specific, we may construct an HPTP map $\Gamma_{B\rightarrow B^n}$ and decompose it into a linear combination of local channels $\cN_j$ for Bob, i.e., $\Gamma_{B\rightarrow B^n}=\sum_{j=1}c_j\cN_j$, where $c_j$ are certain real numbers. One can estimate the shadow information by sampling quantum channels $\cN_j$ and post-processing the measurement outcomes~\cite{Zhao2022} for an observable $O$ and quantum state $\rho_{AB}$ (see Proposition~\ref{prop:UBP} for a precise statement). Then, it is essential to understand the power and limitations of such virtual quantum broadcasting as the following two questions arise:

\begin{enumerate}
    \item \emph{Is there a universal virtual quantum broadcasting protocol?}
    \item \emph{What is the optimal protocol with minimum sampling cost?}
\end{enumerate}

In this paper, we fully address these two questions. In Sec.~\ref{sec:existence}, we demonstrate the existence of a \textit{universal} unilocal virtual $n$-broadcasting protocol, for any bipartite quantum state $\rho_{AB}$ and observable $O$. In Sec.~\ref{sec:opt_vbroad}, we formalize the simulation cost of a unilocal virtual $n$-broadcasting into a semidefinite programming (SDP)~\cite{boyd2004convex}. Notably, we provide an analytical universal unilocal virtual $2$-broadcasting protocol to elucidate the optimal simulation cost.  In addition, we investigate the upper and lower bounds on the simulation cost of the unilocal virtual $n$-broadcasting protocol.

\begin{figure}[t]
    \centering
    \includegraphics[width=0.85\linewidth]{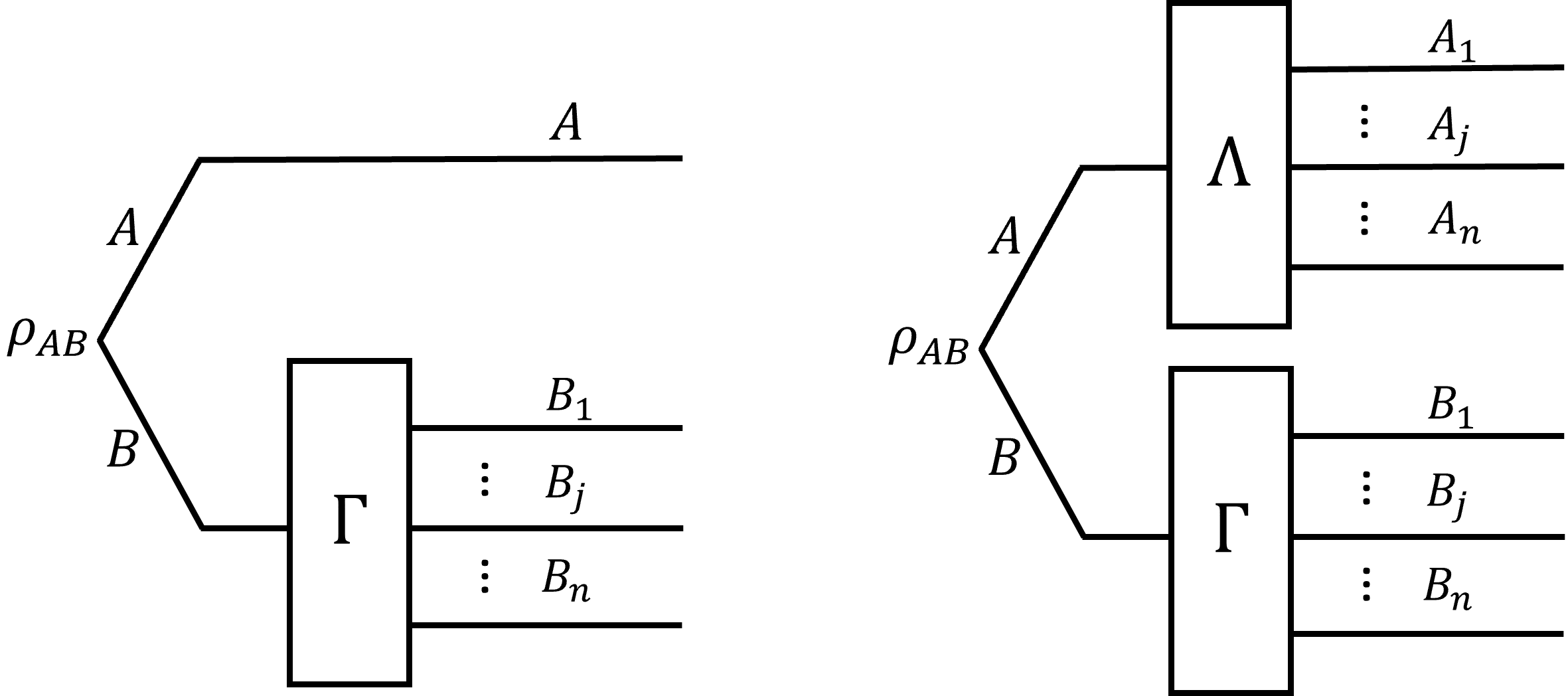}
    \caption{Unilocal(left) and bilocal(right) $n$-broadcasting for bipartite state $\rho_{AB}$. The goal is for the map $\Gamma$ to minimize the dissimilarity between the states on $\rho_{AB_j}$ and $\rho_{AB}$ in a certain measure. Conventionally, $\Gamma$ is a CPTP map, i.e., quantum channel. This paper focuses on the scenario where $\Gamma$ is an HPTP map.
    }
    \label{fig:unilocal}
\end{figure}

\section{Universal virtual broadcasting protocol}\label{sec:existence}

We consider a finite-dimensional Hilbert space $\cH$ and denote $A$ and $B$ as two parties, each possessing their respective Hilbert spaces $\cH_A$ and $\cH_B$. We denote the dimension of  $\cH_B$ as $d$. Let $\{\ket{j} \}_{j=0,\cdots,d-1}$ be a standard computational basis. Denote $\cL(\cH_A)$ as the set of linear operators that map from $\cH_A$ to itself. A linear operator in $\cL(\cH_A)$ is called a density operator if it is positive semidefinite with trace one, and denotes $\cD(\cH_A)$ as the set of all density operators on $\cH_A$. We denote $F_{B_1B_2}:=\sum_{i,j=0}^{d-1}\ketbra{ij}{ji}$ as swap operator between subsystems $B_1$ and $B_2$, and denote $\Phi_{BB_1}:=\sum_{i,j=0}^{d-1} \ketbra{ii}{jj}_{BB_1}$ as the unnormalized $d\ox d$ maximally entangled state. In the absence of ambiguity, subsystems may be omitted, i.e., $\Phi_d$. A quantum channel $\cN_{A\to B}$ is a linear map from $\cL(\cH_A)$ to $\cL(\cH_B)$ that is completely positive and trace-preserving (CPTP). Its associated Choi-Jamiołkowski operator is expressed as $J^{\cN}_{AB} := \sum_{i, j=0}^{d-1}\ketbra{i}{j} \ox \cN_{A \to B}(\ketbra{i}{j})$.

\begin{figure*}[t]
    \centering
    \includegraphics[width=0.6\textwidth]{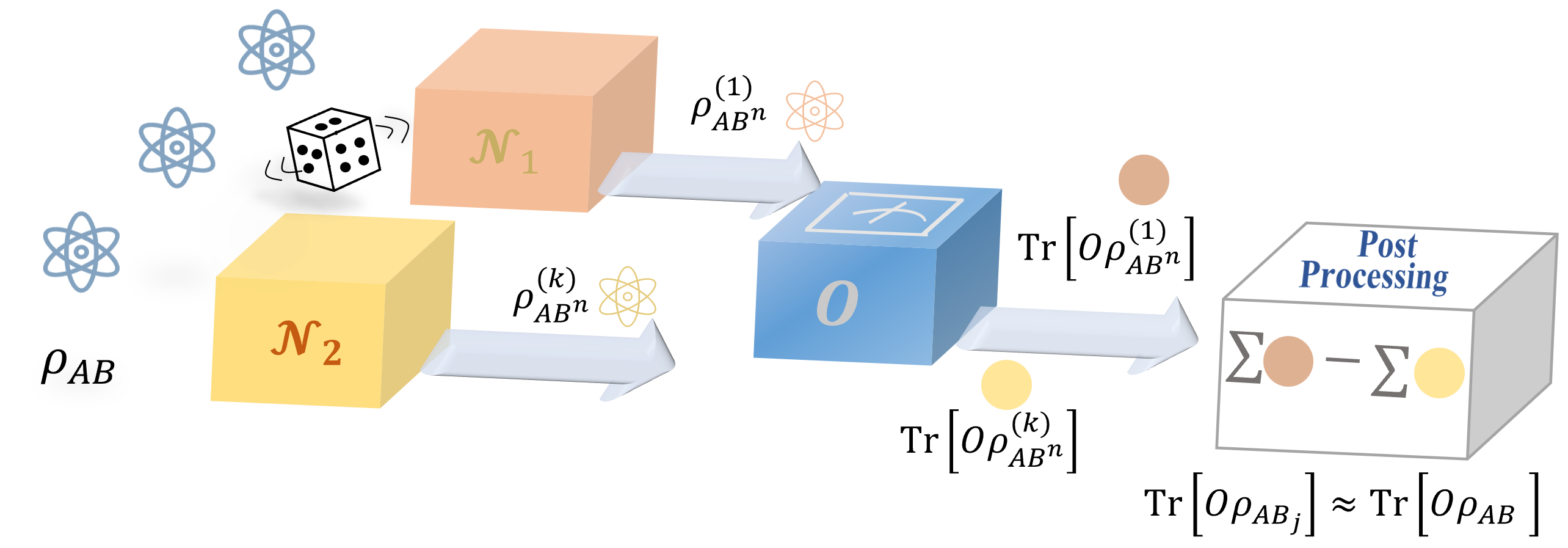}
    \caption{Illustration of using a universal virtual $n$-broadcasting $\Gamma_{B\to B^n}=p_1\cN_1-p_2\cN_2$ to share shadow information between different parties. For a given observable $O$ and many copies of a bipartite state $\rho_{AB}$, we sample local quantum channels $\cN_1$ and $\cN_2$ with probability $p_1/(p_1+p_2)$ and $p_2/(p_1+p_2)$ respectively. Iterating this procedure $m$ times, we can obtain $\rho_{AB^n}^{(k)}$ for $k = 1,2,\cdots,m$. Afterwards, each party $AB_j$, where $j = 1,2,\cdots,n$, obtains $\tr[O\rho_{AB}]$ since $\tr[O\rho_{AB}]=\tr[O\tr_{\backslash AB_j}[\Gamma_{B\to B^n}(\rho_{AB})]]$.}
    \label{fig:VBP}
\end{figure*}

Formally, a unilocal virtual $n$-broadcasting protocol for a bipartite quantum state $\rho_{AB}$ is defined as follows.

\begin{definition}[Unilocal virtual $n$-broadcasting protocol]
For a bipartite state $\rho_{AB}\in\cD(\cH_A\otimes\cH_B)$, an HPTP map $\Gamma_{B\rightarrow B^n}$ is called a unilocal virtual $n$-broadcasting protocol for $\rho_{AB}$ if 
\begin{align}
    \rho_{AB} =\tr_{\setminus AB_j}[\Gamma_{B\rightarrow B^n}(\rho_{AB})], \quad \forall j=1,2,\cdots,n,
\end{align}
where identity map is omitted, $\tr_{\setminus AB_j}$ denotes taking partial trace on the subsystems excluding $AB_j$, and $B^n$ is the abbreviation of the subsystems $B_1B_2\cdots B_n$.
\end{definition}

We note that if there is a unilocal virtual $n$-broadcasting protocol $\Gamma_{B\rightarrow B^n}$ for all quantum states $\rho_{AB}\in\cD(\cH_A\ox \cH_B)$, we call it \textit{a universal unilocal virtual $n$-broadcasting protocol}. Equivalently, a universal unilocal virtual $n$-broadcasting protocol $\Gamma_{B\rightarrow B_1\cdots B_n}$ can be characterized by its Choi operator $J^{\Gamma}_{BB^n}$ as the following Lemma.

\begin{lemma}\label{lem:state to channel}
An HPTP map $\Gamma_{B\rightarrow B^n}$ is a universal unilocal virtual $n$-broadcasting protocol if and only if 
\begin{align}
    J_{BB_j}^{\Gamma}=\Phi_{BB_j}, \quad j=1,\cdots, n,
\end{align}
where $\Phi_{BB_j}$ denotes the unnormalized $d\ox d$ maximally entangled state on system $BB_j$, $J^{\Gamma}_{BB_j} := \tr_{\setminus BB_j}[J^\Gamma_{BB^n}]$, and $J^\Gamma_{BB^n}$ is the Choi operator of $\Gamma_{B\rightarrow B^n}$.
\end{lemma}

Lemma~\ref{lem:state to channel} states that a universal unilocal virtual $n$-broadcasting protocol can be described by its Choi operator, which means we can check constraints on Choi operators instead of constraints involving input and output states. The proof can be found in the Appendix. One of the remarkable and valuable findings in this paper is that there indeed exists a universal virtual $n$-broadcasting protocol. As a warm-up example, we present a universal unilocal virtual $2$-broadcasting protocol as follows:
\begin{equation*}
\begin{aligned}
    \Gamma_{B\rightarrow B_1B_2}(\rho_{AB}) := \rho_{AB_1}\ox \frac{I_{B_2}}{d} &+\cS_{B_1B_2}(\rho_{AB_1}\ox \frac{I_{B_2}}{d})\\
    &-\cR_{B\rightarrow B_1B_2}(\rho_{AB}),
\end{aligned}
\end{equation*}
where $\mathcal{S}_{B_1B_2}(\cdot)$ denotes the swap operation between the subsystem $B_1$ and $B_2$, $\cR_{B\rightarrow B_1B_2}(\cdot)$ denotes the replacement channel yielding the normalized $d\ox d$ maximally entangled state between subsystem $B_1$ and $B_2$ for any input state. Its Choi operator can be written as
\begin{align*}
    J_{BB_1B_2}^{\Gamma_{B\rightarrow B_1B_2}}:=\frac{1}{d}\Phi_{BB_1}\ox I_{B_2}+\frac{1}{d}\Phi_{BB_2}\ox I_{B_1}-\frac{1}{d}\Phi_{B_1B_2}\ox I_{B}.
\end{align*}
It is straightforward to check that $J^{\Gamma_{BB_1B_2}}_{BB_1}=J^{\Gamma_{BB_1B_2}}_{BB_2}=\Phi_{BB_1}=\Phi_{BB_2}$. Consequently, $\Gamma_{B\rightarrow B_1B_2}$ is a universal unilocal virtual $2$-broadcasting protocol by Lemma~\ref{lem:state to channel}. Furthermore, we extend our investigation to encompass the realm of $n$-broadcasting, where we demonstrate the existence of a universal unilocal virtual $n$-broadcasting protocol as follows.

\begin{proposition}\label{prop:UBP}
For any bipartite quantum system $AB$, there exists a universal unilocal virtual $n$-broadcasting protocol.
\end{proposition}

We demonstrate Proposition~\ref{prop:UBP} by explicitly constructing an HPTP map $\Gamma^\prime_{B\rightarrow B^n}$ as follows.
\begin{equation}
\begin{aligned}
    \Gamma^\prime_{B\rightarrow B^n}(\rho_{AB}):=\sum_{j=1}^n \cS_{B_1B_j} & (\rho_{AB_1}\ox \frac{I_{B_2\cdots B_n}}{d^{n-1}})\\
    &-(n-1)\cR_{B\rightarrow B^n}(\rho_{AB}),
\end{aligned}
\end{equation}
where $\mathcal{S}_{B_1B_j}(\cdot)$ denotes the swap operation between the subsystems $B_1$ and $B_j$, and $\cR_{B\rightarrow B^n}(\cdot)$ denotes the replacement channel yielding $\Phi_{B_1B_2}\ox \frac{I_{B_3\cdots B_n}}{d^{n-3}}$ for any input state. By checking its Choi operator and applying Lemma~\ref{lem:state to channel}, we know it is a universal unilocal $n$-broadcasting protocol.

Such a universal unilocal virtual $n$-broadcasting protocol can be implemented via a quasiprobability decomposition strategy~\cite{Jiang2020,Zhao2022,zhao2024retrieving} as shown in Fig.~\ref{fig:VBP}. Given an observable $O$ and $M$ copies of a bipartite state $\rho_{AB}$ shared between Alice and Bob, a unilocal virtual $n$-broadcasting protocol can be decomposed as
$\Gamma_{B\rightarrow B^n} = p_1\cN_1 - p_2\cN_2$, where $\cN_1$ and $\cN_2$ are quantum channels ~\cite{Jiang2020}. 
In $m$-th round of sampling, Bob samples a quantum channel $\cN^{(m)}\in\{\cN_1, \cN_2\}$ with probability $p^{(m)}\in \{p_1/\gamma, p_2/\gamma\}$ where $\gamma= p_1 + p_2$. Then apply the channel to $\rho_{AB}$ obtaining $\rho_{AB^n}^{(m)}$. Repeat this process $M$ times to obtain $M$ copies of state $\{\rho_{AB^n}^{(1)}, \rho_{AB^n}^{(2)}, \cdots, \rho_{AB^n}^{(M)}\}$. Without loss of generality, if a global measurement is performed on a computational basis on $AB_j$, i.e., $O = \sum_{k}\lambda_k \ketbra{k}{k}, \lambda_k\in[-1,1]$, we construct
\begin{equation}\label{Eq:est_trOrho}
    \xi := \frac{\eta}{M}\sum_{m=1}^M\operatorname{sgn}(p^{(m)})\lambda^{(m)}
\end{equation}
as an estimator of $\tr[O\rho_{AB}]$. Subsequently, each bipartite system $AB_j$, where $j = 1,2,\cdots,n$, acquires the information of $\tr[O\rho_{AB}]$ by measuring their subsystems in the eigenbasis of $O$ and post-processing~\cite{Zhao2022}.

Note that for any bipartite system $AB_j$, there are $M$ copies of quantum states $\rho_{AB_j}^{(m)},\, m=1,2,...,M$, each of which is labeled by a classical bit $\mathrm{sgn}(p^{(m)})$ where $p^{(m)} = p_1 \text{ or } p_2$. Instead of directly transmitting the expectation value, $AB_j$ can further apply any further quantum operations to these samples and access the same correlation as $AB$ through measurement statistics.
This universal unilocal virtual $n$-broadcasting protocol can apply to any state $\rho_{AB}$ and any observable $O$. However, it is impossible when one considers using one channel to deal with this task. Consequently, we may extend the no-go theorem for local broadcasting by involving HPTP maps in the broadcasting procedure.

\section{Optimal virtual broadcasting}\label{sec:opt_vbroad}

In this section, we explore the universal unilocal virtual $n$-broadcasting protocol which can be simulated by physical operations with minimum costs. Treating the unilocal virtual $n$-broadcasting protocol as a general HPTP map, its simulation or sampling cost can be characterized via the following physical implementability, which plays the key role in quantifying the number of rounds required to reach the desired estimating precision~\cite{Jiang2020}.

\begin{definition}[Simulation cost of an HPTP map~\cite{Jiang2020}]\label{def:cost_of_hptp}
The simulation cost (or physical implementability) of an HPTP map $\Gamma$ is defined as
\begin{equation}
\begin{aligned}
    \nu(\Gamma):=\log\min\Big\{p_1+&p_2\big| \, \Gamma = p_1\mathcal{N}_1-p_2\mathcal{N}_2,\\
    \ & p_1,p_2\geq 0, \cN_1,\cN_2\in\text{CPTP} \Big\},
\end{aligned}
\end{equation}
where logarithms are in base $2$ throughout this paper.
\end{definition}

By Hoeffding's inequality, denoting $\gamma=p_1+p_2$, it requires at least $\cO(\frac{\gamma^2}{\delta^2}\ln{\frac{2}{\epsilon}})$ samples of $\rho_{AB}$ to achieve the estimation error $\delta$ with a probability $1-\epsilon$, for estimating $\tr(O\rho_{AB})$ by estimator $\xi$ in Eq.~\eqref{Eq:est_trOrho}. Based on the above, we define the optimal simulation cost of a universal unilocal virtual $n$-broadcasting protocol as follows.

\begin{definition}[Optimal simulation cost]\label{def:optimal simulation cost}
    The optimal simulation cost of all universal unilocal virtual $n$-broadcasting protocols is defined as
    \begin{align}
        \gamma^{\ast}_n:=\min\{\nu(\Gamma_{B\to B^n}):\Gamma_{B\to B^n}\in\mathcal{T}_n\},
    \end{align}
    where $\mathcal{T}_n$ denotes the set of all universal unilocal virtual $n$-broadcasting protocols. The corresponding protocol $\Gamma^{\ast}_{B\to B^n}:=\argmin \{\nu(\Gamma_{B\to B^n}):\Gamma_{B\to B^n}\in\mathcal{T}_n\}$ is the optimal universal $n$-broadcasting protocols.
\end{definition}
Combined with the properties that a universal virtual broadcasting should satisfy as stated in Lemma~\ref{lem:state to channel}, the optimal simulation cost can be formalized as follows.

\begin{proposition}\label{prop:SDP-universal}
The optimal simulation cost of all universal unilocal virtual $n$-broadcasting protocols can be characterized as the following SDP:
\begin{equation}\label{eq:primal_without_rho}
\begin{aligned}
    2^{\gamma^{\ast}_n} =  \min\;&p_1+p_2\\
    {\rm s.t.}\; &\tr_{\backslash BB_j}[J^{\cN_{1}}_{BB^n}-J^{\cN_{2}}_{BB^n}]=\Phi_{BB_j},\,j=1,\cdots,n\\
    &\tr_{B^n}[J^{\cN_{1}}_{BB^n}]=p_1I_{B},\\
    &\tr_{B^n}[J^{\cN_{2}}_{BB^n}]=p_2I_{B},\\
    &J^{\cN_{1}}_{BB^n}\geq 0, J^{\cN_{2}}_{BB^n}\geq 0,
\end{aligned}
\end{equation}
with variables $J^{\cN_{1}}_{BB^n}$, $J^{\cN_{2}}_{BB^n}$, $p_1$ and $p_2$. $\Phi_{BB_j}$ is the unnormalized $d\ox d$ maximally entangled state on system $BB_j$.
\end{proposition}
\begin{proof}
     Let $\mathcal{T}_n$ be the set of all universal unilocal virtual $n$-broadcasting protocols, and $J^\Gamma_{BB^n}$ be the Choi operator of $\Gamma_{B\to B^n}\in\mathcal{T}_n$.
     By definition of the simulation cost given in Definition~\ref{def:cost_of_hptp}, there exist $p_1,p_2\geq0$ and $\cN_1,\cN_2\in\text{CPTP}$ such that 
     \begin{equation}
         \begin{aligned}
             \nu(\Gamma_{BB^n})=\log(p_1+p_2).
         \end{aligned}
     \end{equation}
     The Choi operator of $\Gamma_{BB^n}$ satisfies $J^\Gamma_{BB^n}=p_1\hat{J}^{\cN_1}_{BB^n}-p_2\hat{J}^{\cN_2}_{BB^n}$, where $\hat{J}^{\cN_1}_{BB^n}$ and $\hat{J}^{\cN_2}_{BB^n}$ denote the Choi operators of $\cN_1$ and $\cN_2$, respectively. We further rewrite $J^{\cN_1}_{BB^n}:=p_1\hat{J}^{\cN_1}_{BB^n}$ and $J^{\cN_2}_{BB^n}:=p_2\hat{J}^{\cN_2}_{BB^n}$ for simplifying this optimization problem. According to Lemma~\ref{lem:state to channel} and Def.~\ref{def:optimal simulation cost}, one can obtain the SDP in Eq.~\eqref{eq:primal_without_rho}, which completes this proof.
\end{proof}

We further present its dual SDP as follows:  
\begin{equation}\label{eq:dual_without_rho}
\begin{aligned}
\max\;& \sum_{j=1}^n\tr[X_{BB_j}\Phi_{B{B}_j}]\\
        {\rm s.t.}\;& \tr[Z_{B}]\leq 1, \tr[K_{B}]\leq 1,\\
        & Z_{B}\ox I_{B^n} - \sum_{j=1}^n \cS_{B_1B_j}(X_{BB_1}\ox I_{B_2\cdots B_n})\geq 0,\\
        & K_{B}\ox I_{B^n} + \sum_{j=1}^n \cS_{B_1B_j}(X_{BB_1}\ox I_{B_2\cdots B_n})\geq 0,\\
        & (j=1,\cdots,n),
\end{aligned}
\end{equation}
where $X_{BB_j}$, $Z_B$ and $K_B$ are optimization variables, and $\cS_{B_1B_j}$ denotes the swap operator between system $B$ and $B_j$. We retain the derivation in Appendix~\ref{appendix:universal_dual_sdp}. 

The above SDPs allow us to explore the optimal universal virtual broadcasting protocols that can achieve the optimal simulation cost.  
Specifically, we give the analytical optimal simulation cost for a universal unilocal virtual $2$-broadcasting and obtain the optimal universal protocol.

\begin{theorem}
[Optimal simulation cost of virtual $2$-broadcasting]
\label{thm:cost_2broadcast_qudit}
The optimal simulation cost of all universal unilocal virtual 2-broadcasting protocols which broadcast system $B$ to $B_1B_2$ is given by
\begin{align}
    \gamma^{\ast}_2=\log\left(3-\frac{4}{d+1}\right),
\end{align}
where $d$ denotes the dimension of Hilbert spaces $\cH_{B}$, $\cH_{B_1}$ and $\cH_{B_2}$.
\end{theorem}

\begin{proof}
First, we are going to prove $2^{\gamma^{\ast}_2} \leq \frac{3d-1}{d+1}$ using the primal SDP in Eq.~\eqref{eq:primal_without_rho}. Denoting $M=\Phi_{BB_1}\ox I_{B_2}$ and $N=I_{B}\ox F_{B_1B_2}$, we shall show that $\{p_1,p_2,J^{\cN_1},J^{\cN_2}\}$ is a feasible solution, where $p_1=\frac{2d}{d+1}$, $p_2=\frac{d-1}{d+1}$, and
\begin{equation}\label{eq:cptp_choi}
\begin{aligned}
    &J^{\cN_1}:=\frac{M+NMN+MN+NM}{2(d+1)},\\
    &J^{\cN_2}:=\frac{1}{d^2-2}\left(I-\frac{d(M+NMN)-(MN+NM)}{d^2-1}\right),
\end{aligned}
\end{equation}
respectively. It is straightforward to check the equality constraints in Eq.~\eqref{eq:primal_without_rho} hold. For the inequality constraints, we find that $(J^{\cN_1})^2=J^{\cN_1}$ and $(J^{\cN_2})^2=\frac{1}{d^2-2}J^{\cN_2}$. Thus $1$ and $\frac{1}{d^2-2}$ are unique non-negative eigenvalues of $J^{\cN_1}$ and $J^{\cN_2}$, respectively, which means $J^{\cN_1}\geq 0$ and $J^{\cN_2}\geq 0$. Therefore, $\{p_1J^{\cN_1}, p_2J^{\cN_2}\}$ is a feasible solution with the cost of $\frac{3d-1}{d+1}$, which implies $2^{\gamma^{\ast}_2} \leq \frac{3d-1}{d+1}$. 

Second, we use dual SDP in Eq.~\eqref{eq:dual_without_rho} to show that $2^{\gamma^{\ast}_2} \geq \frac{3d-1}{d+1}$. We show that $\{X_{BB_1},Y_{BB_2},Z_B,K_B\}$ is a feasible solution, where 
\begin{equation}
    Z_B=K_B=\frac{1}{d}I \text{ and } X_{BB_1} = Y_{BB_2} = \frac{2}{d(d+1)}\Phi_d - \frac{1}{2d}I.
\end{equation}
Still, we can check that $\{X_{BB_1},Y_{BB_2},Z_B,K_B\}$ satisfies the constraints SDP in Eq.~\eqref{eq:dual_without_rho}. Specifically, we have $I-\frac{M+NMN}{d+1}\geq0$ since $\frac{M+NMN}{d+1}$ is a Hermitian matrix with maximal eigenvalue one~\cite{leditzky2017useful}. Therefore, $\{X_{BB_1},Y_{BB_2},Z_B,K_B\}$ is a feasible solution. Finally, we further check the objective function,
\begin{align}
    \tr[X_{BB_1}\Phi_{BB_1}]+\tr[Y_{BB_2}\Phi_{BB_2}] = 3-\frac{4}{d+1},
\end{align}
which yields $2^{\gamma^{\ast}_2}\geq \frac{3d-1}{d+1}$. Combining the primal part and the dual part, we conclude that 
\begin{align}
    \gamma^{\ast}_2=\log\left(3-\frac{4}{d+1}\right),
\end{align}
which completes this proof.
\end{proof}

\begin{proposition}[Optimal universal $2$-broadcasting protocol]\label{prop:opt-2}
The optimal universal $2$-broadcasting protocol is given by $\Gamma^\ast_{B\to B_1B_2}=p_1\cN_1-p_2\cN_2$, where $p_1=\frac{2d}{d+1}$, $p_2=\frac{d-1}{d+1}$, and
    \begin{equation}
    \begin{aligned}
        \cN_1(\rho_{AB}):=&\frac{d}{d+1}\cP(\rho_{AB})+\frac{1}{d+1}\mathcal{Q}(\rho_{AB}),\\
        \cN_2(\rho_{AB}):=&\frac{d^2}{d^2-2}\mathcal{I}(\rho_{AB})-\frac{2d^2}{(d^2-2)(d^2-1)}\cP(\rho_{AB})\\
        &\quad\quad+\frac{2}{(d^2-2)(d^2-1)}\mathcal{Q}(\rho_{AB}), 
    \end{aligned}
    \end{equation}
   where $d$ denotes the dimension of Hilbert spaces $\cH_{B}$, $\cH_{B_1}$ and $\cH_{B_2}$, $\mathcal{Q}(\rho_{AB}):=\frac{1}{2}[F_{B_1B_2}(\rho_{AB_1}\otimes I_{B_2})+(\rho_{AB_1}\otimes I_{B_2})F_{B_1B_2}]$, $\cP(\rho_{AB}):=\frac{1}{2d}[\rho_{AB_1}\otimes I_{B_2}+\cS_{B_1B_2}(\rho_{AB_1}\otimes I_{B_2})]$, 
   $\cS_{B_1B_2}(\cdot)$ is the swap operation between $B_1$ and $B_2$ corresponding to the swap operator $F_{B_1B_2}:=\sum_{i,j=1}^{d-1}\ketbra{ij}{ji}$, and $\mathcal{I}(\cdot)$ denotes the replacement channel yielding $\frac{1}{d^2}I_{B_1B_2}$.
\end{proposition}
\begin{proof}
    We further denote $M:=\Phi_{BB_1}\ox I_{B_2}$ and $N:=I_{B}\ox F_{B_1B_2}$ for short. Based on the proof of Theorem~\ref{thm:cost_2broadcast_qudit}, one can find that there exists a virtual $2$-broadcasting protocol $\Gamma^\ast_{B\to B_1B_2}=p_1\cN_1-p_2\cN_2$ with the optimal simulation cost, where $p_1=\frac{2d}{d+1}$, $p_2=\frac{d-1}{d+1}$, $\cN_1$ and $\cN_2$ are quantum channels with Choi operators $J^{\cN_1}:=\frac{M+NMN+MN+NM}{2(d+1)}$ and $J^{\cN_2}:=\frac{1}{d^2-2}\left(I-\frac{d(M+NMN)-(MN+NM)}{d^2-1}\right)$, respectively. According to the statement of Definition~\ref{def:optimal simulation cost}, we can refer to $\Gamma^\ast_{B\to B_1B_2}$ as the optimal universal $2$-broadcasting protocol, which completes this proof.
\end{proof}

Theorem~\ref{thm:cost_2broadcast_qudit} proposes the optimal universal virtual 2-broadcasting protocol, taking into account the sampling cost required to broadcast the correlation inherent in the classical information $\tr[O\rho_{AB}]$ with a desired estimating precision. Note that what we obtained here is the minimum cost protocol among all possible universal unilocal virtual $2$-broadcasting protocols. We first find the HPTP protocol for the desired simulation cost and then utilize the dual SDP in Eq.~\eqref{eq:primal_without_rho} to establish the optimality of this protocol. Moreover, Theorem~\ref{thm:cost_2broadcast_qudit} reveals an intriguing relationship between the sampling cost and the system's dimension. As the dimension of the quantum system grows, the simulation cost for universal virtual 2-broadcasting converges to a constant value of $\log3$, which means that even in high-dimensional quantum systems, the simulation cost is still within a controllable range.

We further extend our investigation to the context of unilocal virtual $n$-broadcasting, to analyze the change in simulation cost in relation to the number of parties involved, i.e., from system $B$ to $B_1,\cdots, B_n$. In particular, we derive an upper bound and a lower bound for the simulation cost of universal virtual $n$-broadcasting.

\begin{theorem}[Upper and lower bounds]\label{thm:bound on n-broadcasting}
    The optimal simulation cost of all universal unilocal virtual n-broadcasting protocols which broadcast system $B$ to $B_1\cdots B_n$ satisfies
    \begin{align}
       \log\left(\frac{2nd}{n+d-1}-1\right) \leq \gamma^{\ast}_n\leq \log(2n-1),
    \end{align}
    where $d$ denotes the dimension of Hilbert spaces $\cH_{B}$ and $\cH_{B_j}$ for $j=1,\cdots, n$.
\end{theorem}

\begin{proof}
    We first show the upper bound on the minimum simulation cost. According to Proposition~\ref{prop:UBP}, one can find that the simulation cost of the universal protocol $\Gamma^\prime_{B\to B^n}$ can be an upper bound of $\gamma^{\ast}_n$, i.e., $\gamma^{\ast}_n\leq \nu(\Gamma^\prime_{B\to B^n})$. Then, rewrite the universal virtual $n$-broadcasting protocol $\Gamma^\prime_{B\rightarrow B^n}$ into the linear combination of two quantum channels $\cM_1$ and $\cM_2$ as
    \begin{align}
        \Gamma^{\prime}_{B\to B^n}=n\cM_1 - (n-1)\cM_2,
    \end{align}
    where the Choi operators of $\cM_1$ and $\cM_2$ can be written as
    $J^{\cM_1}:=\frac{1}{nd^{n-1}}\sum_{j=1}^n \cS_{B_1B_j}(\Phi_{BB_1}\ox I_{B_2\cdots B_n})$ and $J^{\cM_2}:=\frac{1}{d^{n-1}} \Phi_{B_1B_2}\ox I_{BB_3\cdots B_n}$, respectively. Then, by definition, we have $\nu(\Gamma^{\prime}_{B\to B^n})\leq \log(2n-1)$, which directly gives $\gamma^{\ast}_n\leq \log(2n-1)$.
    
    Second, we are going to derive the lower bound by showing that $\{X_{BB_1},\cdots,X_{BB_n},Z_{B},K_{B}\}$ is a feasible solution of the dual SDP in Eq.~\eqref{eq:primal_without_rho}, where $Z_{B}=K_{B}=\frac{I_B}{d}$, and $ X_{BB_1}=\cdots=X_{BB_n}=\frac{2}{d(n+d-1)}\Phi_d-\frac{1}{nd}I$.
It is straightforward to check that $\{X_{BB_1},\cdots,X_{BB_n},Z_{B},K_{B}\}$ satisfies the constrains of SDP in Eq.~\eqref{eq:primal_without_rho}. We further check the objective function
\begin{align}
    \sum_{j=1}^n\tr[X_{BB_j}\Phi_{BB_j}]=\frac{2nd}{n+d-1}-1.
\end{align}
According to the fact that the optimal solution of dual SDP in Eq.~\eqref{eq:primal_without_rho} is a lower bound of the optimal solution of primal SDP in Eq.~\eqref{eq:primal_without_rho}, we have the following inequality
\begin{align}
    \log\left(\frac{2nd}{n+d-1}-1\right)\leq \gamma^{\ast}_n,
\end{align}
which completes the proof.
\end{proof}

\begin{figure}[t]
    \centering
    \includegraphics[width=0.9\linewidth]{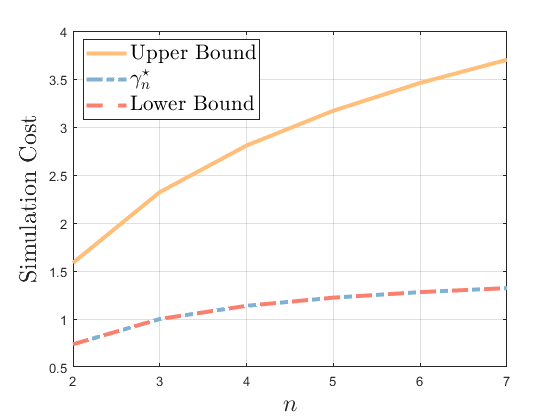}
    \caption{Simulation cost of universal unilocal virtual $n$-broadcasting. Here, the dimensions of the Hilbert spaces $\mathcal{H}_B$ and $\mathcal{H}_{B_j} (j=1,\cdots,n)$  are all equal to two. The $x$-axis corresponds to the number of parties on Bob's side involved in the broadcasting. The $y$-axis corresponds to the simulation cost of the protocol. 
    }
    \label{fig:up_low_cost}
\end{figure}

Remarkably, in Fig.~\ref{fig:up_low_cost}, one can find that the lower bound matches the optimal simulation cost $\gamma_{n}^\ast$ in numerical experiments. Furthermore, according to Theorem~\ref{thm:bound on n-broadcasting}, it is straightforward to find that the lower bound converges to the upper bound as the dimension of the quantum system grows. These mean the upper bound is significantly valuable at a high system level. The simulation cost will not exhibit exponential growth with the dimension of the system, which suggests our capability to effectively tackle the unilocal virtual $n$-broadcasting task, even for a bipartite system with a high dimension. In summary, Theorem~\ref{thm:cost_2broadcast_qudit} and Theorem~\ref{thm:bound on n-broadcasting} reveal that engaging in virtual $n$-broadcasting not only enables the acquisition of correlation of a bipartite quantum system encoded in the expectation values but also grants control over the associated costs. These remind us that it is feasible to employ non-physical operations to overcome the limitations of quantum mechanics at an acceptable cost.

\section{Concluding remarks}
In this work, we have proposed a novel framework known as \textit{unilocal virtual quantum broadcasting}, employing Hermitian-preserving trace-preserving (HPTP) maps. We have demonstrated the existence of a universal unilocal virtual $n$-broadcasting protocol capable of distributing information from any bipartite quantum state to multiple parties via local operations. Furthermore, we have formalized the simulation cost of this broadcasting protocol as a semidefinite programming problem. Notably, we have provided an analytical universal unilocal virtual 2-broadcasting protocol to clarify the optimal simulation cost. By accurately characterizing simulation cost, we found that virtual 2-broadcasting remains applicable in high-dimensional quantum systems, as the corresponding simulation cost converges to a constant $\log 3$ with increasing dimensions. Furthermore, we have provided upper and lower bounds on the simulation cost of the virtual $n$-broadcasting protocol and demonstrated that the lower bound converges to the upper bound $\log(2n-1)$ that is independent of the system dimension. The findings above demonstrate the practical potential of our virtual broadcasting protocol, as the simulation costs are always controllable. It is worth noting that Parzygnat et al.~\cite{parzygnat2024virtual} have explored broadcasting tasks via a virtual process. They focused on the conditions that virtual quantum broadcasting maps should fulfill and provided physical interpretations for their canonical quantum broadcasting map from multiple perspectives. Our work generalizes the virtual broadcasting~\cite{parzygnat2024virtual} to virtual unilocal broadcasting by allowing a reference system
and shows that unilocal virtual broadcasting maps can efficiently accomplish the broadcasting task through simulated physical operations with minimal cost. It is noteworthy that if system $A$ is trivial, the canonical virtual broadcasting map presented in Ref.~\cite{parzygnat2024virtual} is also a feasible 2-broadcasting protocol but not the one with the minimum simulation cost. Specifically, one can check that the simulation cost of the canonical virtual broadcasting map is $\log d$.

Our results open new avenues for understanding and harnessing the unique properties of quantum mechanics. This demonstrates the possibility of overcoming the limitations of quantum mechanics using controllable non-physical operations. The exploration of virtual broadcasting not only broadens our comprehension of quantum information distribution~\cite{murao2000quantum,chuan2012quantum,streltsov2012quantum} but also provides a valuable tool for advancing quantum communication and computing technologies~\cite{bennett1993teleporting}. Future work will focus on the implementation of quantum circuits of our proposed virtual broadcasting protocol and its further practical applications in the areas of quantum communication and computing.

\vspace{2mm}
\textit{Note added.}
While finishing the arXiv preprint version of this work in October 2023, we became aware of the arXiv preprint of a closely related work~\cite{parzygnat2024virtual} that independently proposed the idea of {virtual broadcasting} in October 2023. The main distinction is discussed in the above concluding remarks.

\section*{Acknowledgements}
We would like to thank Ranyiliu Chen and Xuanqiang Zhao for their helpful comments. We also thank the anonymous reviewers for their helpful suggestions, which helped us improve the manuscript. This work was partially supported by the National Key R\&D Program of China (Grant No. 2024YFE0102500), the Guangdong Provincial Quantum Science Strategic Initiative (Grant No. GDZX2303007), the Guangdong Provincial Key Lab of Integrated Communication, Sensing and Computation for Ubiquitous Internet of Things (Grant No. 2023B1212010007), the Start-up Fund (Grant No. G0101000151) from HKUST (Guangzhou), and the Education Bureau of Guangzhou Municipality.

\bibliography{main}

\appendix
\setcounter{subsection}{0}
\setcounter{table}{0}
\setcounter{figure}{0}

\vspace{3cm}

\begin{center}
\large{\textbf{Appendix for} \\ \textbf{
Optimal unilocal virtual quantum broadcasting}}
\end{center}

\renewcommand{\theequation}{A\arabic{equation}}
\renewcommand{\theproposition}{A\arabic{proposition}}
\renewcommand{\thedefinition}{A\arabic{definition}}
\renewcommand{\thefigure}{A\arabic{figure}}
\setcounter{equation}{0}
\setcounter{table}{0}
\setcounter{section}{0}
\setcounter{proposition}{0}
\setcounter{definition}{0}
\setcounter{figure}{0}

\section{The proof of Lemma~\ref{lem:state to channel} and Proposition~\ref{prop:UBP}}

\renewcommand\theproposition{\ref{lem:state to channel}}
\setcounter{proposition}{\arabic{proposition}-1}
\begin{lemma}\label{appendix:choi_feature_for_broadcasting_map}
An HPTP map $\Gamma_{B\rightarrow B^n}$ is a universal unilocal virtual $n$-broadcasting protocol if and only if 
\begin{align}
    J_{BB_j}^{\Gamma}=\Phi_{BB_j}, \quad j=1,\cdots, n,
\end{align}
where $\Phi_{BB_j}$ denotes the unnormalized $d\otimes d$ maximally entangled state on system $BB_j$, $J^{\Gamma}_{BB_j} := \tr_{\setminus BB_j}[J^\Gamma_{BB^n}]$, and $J^\Gamma_{BB^n}$ is the Choi operator of $\Gamma_{B\rightarrow B^n}$.
\end{lemma}
\begin{proof}
Considering the `if part', for $\forall j\in\{1,\cdots,n\}$, $\tr_{\setminus BB_j}[J_{BB^n}^{\Gamma}]=\Phi_{BB_j}$ implies that
\begin{align}
\tr_{B}[\rho_{AB}^{T_{B}}\tr_{\setminus BB_j}[J_{BB^n}^{\Gamma}]]=\rho_{AB},
\end{align}
for all states $\rho_{AB}\in \cD(\cH_A\ox\cH_B)$.
For the `only if' part, we assume $\Gamma_{B\to B^n}$ is a universal unilocal virtual $n$-broadcasting protocol. Then for any input state $\rho_{AB}\in\cD(\cH_A\ox \cH_B)$, we have 
\begin{align}
\rho_{AB}=\tr_{\setminus AB_j}[\rho_{AB}^{T_{B}}J_{BB^n}^{\Gamma}]=\tr_{B}[\rho_{AB}^{T_{B}}\tr_{\setminus BB_j}[J_{BB^n}^{\Gamma}]],
\end{align}
which means $\tr_{\setminus BB_j}[J_{BB^n}^{\Gamma}]$ is a Choi operator of the identity operator from $B$ to $B_j$, i.e., $\tr_{\setminus BB_j}[J_{BB^n}^{\Gamma}]=\Phi_{BB_j}$. Thus, we complete the proof.
\end{proof}

\renewcommand\theproposition{\ref{prop:UBP}}
\setcounter{proposition}{\arabic{proposition}-1}
\begin{proposition}
For any bipartite quantum system $AB$, there exists a universal unilocal virtual $n$-broadcasting protocol.
\end{proposition}
\begin{proof}
The universal unilocal virtual $n$-broadcasting protocol $\Gamma^\prime_{B\rightarrow B^n}$ can be written as
\begin{equation}
\begin{aligned}
    \Gamma^\prime_{B\rightarrow B^n}(\rho_{AB}):=\sum_{j=1}^n &\cS_{B_1B_j} (\rho_{AB_1}\ox \frac{I_{B_2\cdots B_n}}{d^{n-1}})\\
    &-(n-1)\cR_{B\rightarrow B^n}(\rho_{AB}),
\end{aligned}
\end{equation}
where $\mathcal{S}_{B_1B_j}(\cdot)$ denotes the swap operation between the subsystems $B_1$ and $B_j$, and $\cR_{B\rightarrow B^n}(\cdot)$ denotes the replacement channel yielding $\Phi_{B_1B_2}\ox \frac{I_{B_3\cdots B_n}}{d^{n-3}}$ for any input state.

Then the Choi operator of $\Gamma^\prime_{B\rightarrow B^n}$ is 
\begin{equation}
\begin{aligned}
    J_{BB^n}^{\Gamma^\prime}:=\frac{1}{d^{n-1}}&\cS_{B_1B_j}(\sum_{j=1}^n \Phi_{BB_1}\ox I_{B_2\cdots B_n})\\
    &-\frac{n-1}{d^{n-1}} \Phi_{B_1B_2}\ox I_{BB_3\cdots B_n},\label{eq:universal choi}
\end{aligned} 
\end{equation}
Then, it is straightforward to check that
\begin{align}
    \tr_{\setminus BB_j}[J_{BB^n}^{\Gamma^\prime}]=\Phi_{BB_j},\quad j=1,\cdots,n
\end{align}
Hence we conclude that the HPTP map with Choi operator $J_{BB^n}^{\Gamma^\prime}$ achieves $n$-broadcasting for all states $\rho_{AB}$ by Lemma~\ref{appendix:choi_feature_for_broadcasting_map}.
\end{proof}

\section{SDP for unilocal virtual broadcasting protocol}

\label{appendix:universal_dual_sdp}
\renewcommand\theproposition{\ref{prop:SDP-universal}}
\setcounter{proposition}{\arabic{proposition}-1}

\begin{proposition}
The optimal simulation cost of all universal unilocal virtual $n$-broadcasting protocols can be characterized as the following SDP:
\begin{equation}\label{appendix:primal_without_rho}
\begin{aligned}
    2^{\gamma^{\ast}_n} =  \min\;&p_1+p_2\\
    {\rm s.t.}\; &\tr_{\backslash BB_j}[J^{\cN_{1}}_{BB^n}-J^{\cN_{2}}_{BB^n}]=\Phi_{BB_j},\,j=1,\cdots,n\\
    &\tr_{B^n}[J^{\cN_{1}}_{BB^n}]=p_1I_{B},\\
    &\tr_{B^n}[J^{\cN_{2}}_{BB^n}]=p_2I_{B},\\
    &J^{\cN_{1}}_{BB^n}\geq 0, J^{\cN_{2}}_{BB^n}\geq 0,
\end{aligned}
\end{equation}
with variables $J^{\cN_{1}}_{BB^n}$, $J^{\cN_{2}}_{BB^n}$, $p_1$ and $p_2$. $\Phi_{BB_j}$ is the unnormalized $d\otimes d$ maximally entangled state on system $BB_j$.
\end{proposition}
Now we derive its dual SDP for the case of 2-broadcasting. Based on the primal SDP, the Lagrange function can be written as
\begin{align}
    \mathcal{L}&(X_{BB_1}, Y_{BB_2},Z_{B}, K_{B},J^{\cN_1},J^{\cN_2},p_1,p_2)\\
    &:= p_1+p_2 + \langle X_{BB_1},\,\Phi_{BB_1}-\tr_{B_2}[J^{\cN_{1}}-J^{\cN_{2}}]\rangle\nonumber\\
    &\quad+\langle Y_{BB_2},\,\Phi_{BB_2}-\tr_{B_1}[J^{\cN_{1}}-J^{\cN_{2}}]\rangle\nonumber\\
    &\quad+\langle Z_{B},\, p_1I_{B}-\tr_{B_1B_2}[J^{\cN_1}]\rangle\nonumber\\
    &\quad+\langle K_{B},\, p_2I_{B}-\tr_{B_1B_2}[J^{\cN_2}]\rangle\\
    &=\tr[X_{BB_1}\Phi_{BB_1}]+\tr[Y_{BB_2}\Phi_{BB_2}] \nonumber\\
    &\quad+p_1(\tr[Z_B]+1)+p_2(\tr[K_B]+1)\nonumber\\
    &\quad+\langle -Z_B\ox I_{B_1B_2},\, J^{{\cN}_1}\rangle +\langle -K_B\ox I_{B_1B_2},\, J^{{\cN}_2}\rangle\nonumber\\
    &\quad+ \langle -X_{BB_1}\ox I_{B_2}, \, J^{{\cN}_1}-J^{{\cN}_2}\rangle\nonumber\\
    &\quad+ \langle -Y_{BB_2}\ox I_{B_1}, \, J^{{\cN}_1}-J^{{\cN}_2}\rangle,
\end{align} 
where $X_{BB_1}$, $Y_{BB_2}$, $Z_{B}$ and $K_{B}$ are Lagrange multipliers. Then, the Lagrange dual function can be written as
\begin{equation}
\begin{aligned}
     \mathcal{G}&(X_{BB_1}, Y_{BB_2},Z_{B},K_{B})\\
     &:=\inf_{J^{{\cN}_1}\geq0,J^{{\cN}_2}\geq 0,p_1,p_2}\mathcal{L}(p_1,p_2,X_{BB_1}, Y_{BB_2},\\
     &\qquad\qquad\qquad\qquad\qquad Z_{B},K_{B},J^{{\cN}_1},J^{{\cN}_2}).
\end{aligned}
\end{equation}
Since $J^{{\cN}_1}\geq0$ and $J^{{\cN}_2}\geq0$, it holds that $\tr [Z_B]\geq -1$, $\tr [K_B] \geq -1$,
\begin{align*}
    &-Z_B\ox I_{B_1B_2} - (X_{BB_1}\ox I_{B_2}+Y_{BB_2}\ox I_{B_1}) \geq 0,\\
    &-K_B\ox I_{B_1B_2} + (X_{BB_1}\ox I_{B_2}+Y_{BB_2}\ox I_{BB_1}) \geq 0,
\end{align*}
otherwise, the inner norm is unbounded. Redefine $Z_B$ as $-Z_B$ and $K_B$ as $-K_B$. Then, we obtain the following dual SDP.
\begin{equation}
\begin{aligned}\label{appendix:dual_without_rho_2broadcast}
\max& \;\;\tr[X_{BB_1}\Phi_{BB_1}]+\tr[Y_{BB_2}\Phi_{BB_2}]\\
    {\rm s.t.} &\;\;\tr [Z_B]\leq 1,\\
            & \;\;\tr [K_B]\leq 1,\\
            & \;\;Z_B\ox I_{B_1B_2} - X_{BB_1}\ox I_{B_2} - Y_{BB_2}\ox I_{B_1} \geq 0,\\
            & \;\;K_B\ox I_{B_1B_2} + X_{BB_1}\ox I_{B_2} + Y_{BB_2}\ox I_{B_1} \geq 0.
\end{aligned}
\end{equation}
It is worth noting that the strong duality is held by Slater's condition. Similarly, it is straightforward to generalize it to the case of $n$-broadcasting shown in Eq.~\eqref{eq:dual_without_rho}.

\end{document}